\RequirePackage{fix-cm}
\documentclass[smallcondensed]{svjour3}     % onecolumn (ditto)
\smartqed  % flush right qed marks, e.g. at end of proof
\usepackage{mathptmx}  %Change the fonts. Required !
\usepackage{graphicx}
\usepackage{amsmath, amssymb}
\usepackage{inputenc,verbatim, enumerate}
\usepackage{float}
\usepackage{tikz,pgfplots}
\usetikzlibrary{matrix, fit}
\usetikzlibrary{backgrounds}
\usetikzlibrary{arrows.meta}
\usepackage{url}
\usepackage{cite}
\usepackage{xcolor,color}

\usepackage{algorithm}% http://ctan.org/pkg/algorithms
\usepackage{algpseudocode}

\usepackage{xfrac}

\DeclareMathAlphabet{\mathcal}{OMS}{cmsy}{m}{n}

\spnewtheorem{construction}{Construction}{\bf}{\rm}
\newcommand{\vc}{\mathbf{c}}
\newcommand{\F}{\mathbb{F}}

\newcommand{\Z}{\mathbb{Z}}

\newcommand{\GG}{\mathcal{G}}
\newcommand{\FF}{\mathcal{F}}
\newcommand{\HH}{\mathcal{H}}
\newcommand{\YY}{\mathcal{Y}}

\newcommand{\CC}{\mathcal{C}}
\newcommand{\Ss}{\mathcal{S}}

\newcommand{\cl}{\mathrm{cl}}
\newcommand{\opt}{\mathrm{opt}}
\newcommand{\ie}{\emph{i.e.}}

\newcommand{\supp}{\mathrm{supp}}
\newcommand{\wt}{\mathrm{wt}}

\journalname{Designs, Codes and Cryptography}
\begin{document}

\title{The complete hierarchical locality of the punctured simplex code
	\thanks{Part of the results were submitted without any proofs to IEEE International Symposium Information Theory (ISIT) 2019.}
}

\author{Matthias Grezet \and Camilla Hollanti
}

\institute{Matthias Grezet \and Camilla Hollanti \at
              Department of Mathematics and Systems Analysis, Aalto University, Finland \\
            \email{matthias.grezet@aalto.fi, camilla.hollanti@aalto.fi} 
}

\date{Received: date / Accepted: date}

\maketitle

\begin{abstract}
This paper presents a new alphabet-dependent bound for codes with hierarchical locality. Then, the complete list of possible localities is derived for a class of codes obtained by deleting specific columns from a Simplex code. This list is used to show that these codes are optimal codes with hierarchical locality. 
\end{abstract}

\section{Introduction}
%Introduction

%Contains : 

In modern distributed storage systems (DSSs) failures happen frequently, whence decreasing the number of connections required for node repair is crucial. Locally repairable codes (LRCs) are a subclass of erasure-correcting codes, which allow a small number of failed nodes to be repaired by accessing only a few other nodes. LRCs were introduced in \cite{gopalan12}, \cite{papailiopoulos12} where the codes can locally repair one failure. They  were later extended in \cite{prakash12}, \cite{kamath13} to be able to locally repair up to $\delta-1$ failures. 

An $[n,k,d]$ linear code $\CC$ of length $n$, dimension $k$, and minimum hamming distance $d$, has all-symbol locality $(r,\delta)$ if for all code symbols $i \in [n]=\{1, \ldots , n \}$, there exists a set $R_{i} \subseteq [n]$ containing $i$ such that $|R_{i}| \leq r + \delta -1$ and the minimum distance of the restriction of $\CC$ to $R_{i}$ is at least $\delta$. We refer to $\CC$ as an $(n,k,d,r,\delta)$-LRC and to the sets $R_{i}$ as repair sets or local sets. Related Singleton-type bounds have been derived for various cases in \cite{gopalan12,papailiopoulos12,prakash12} and the first bound with a fixed code alphabet was obtained in \cite{cadambe15} for $\delta=2$. Constructions achieving the Singleton-type bounds and the bound in \cite{cadambe15} for $\delta=2$ were provided in \cite{papailiopoulos12,prakash12, kamath13, rawat16, kamath13b, rawat14b, tamo14, westerback15, silberstein18, zeh15}. 

The authors of \cite{agarwal18} proposed the first alphabet-dependent bound on LRCs over the alphabet $Q$ using an upper bound $B(n,d)$ on the cardinality of a code of length $n$ and minimum distance $d$. The global bound is as follows:  
\begin{equation}
\label{eq:ABHMT}
k \leq \left( \left\lceil \frac{n-d+1}{r+\delta-1} \right\rceil +1 \right) \log_{q} B(r+\delta-1,\delta).
\end{equation}

Recently, \cite{grezet18LRC} provided a different alphabet-dependent bound for LRCs of the same type as the bound in \cite{cadambe15} using the Griesmer bound $\GG_{q}(k,d):=\sum\limits_{i=0}^{k-1} \lceil d / q^{i} \rceil$. The bound has the following form : For any linear $(n,k,d,r,\delta)$-LRC $\CC$ with $\kappa$ the upper bound on the local dimension of the repair sets, 
\begin{equation}
\label{eq:CMG_r}
k \leq \min_{\lambda \in \Z_{+}} \left\{ \lambda + k_{\opt}^{(q)}(n-\mu, d ) \right\}
\end{equation}
where $a,b \in \Z$ such that $\lambda = a \kappa + b, 0 \leq b < \kappa$ and $\mu = (a+1) \GG_{q}(\kappa, \delta) - \GG_{q}(\kappa - b, \delta) $.

%H-LRC--------
In \cite{sasidharan15}, the authors  introduced the notion of codes with hierarchical locality (H-LRCs), which optimizes futher the number of nodes contacted for repair according to the number of failures. A 2-level H-LRC is a code where the restrictions to the repair sets are themselves LRCs, thus providing an extra layer of locality. If an H-LRC has locality $[(r_{1},\delta_{1}),(r_{2},\delta_{2})]$, then the number of nodes contacted to repair up to $\delta_{2}-1$ failures is at most $r_{2}$ and the number of nodes contacted for repair is at most $r_{1}$ if the number of failures is $\geq \delta_{2}$ and $\leq \delta_{1}-1$. This concept can be easily generalized to an arbitrary level of hierarchy. A Singleton-type bound for codes with hierarchical locality was derived in \cite{sasidharan15} and constructions attaining the bound were provided in \cite{sasidharan15, ballentine18}.

In Section \ref{sec:bound}, we show how we can adapt the construction algorithms provided in \cite{sasidharan15} to obtain an alphabet-dependent bound for H-LRCs of the same type as in \cite{cadambe15}. By construction, this bound is at least as good as the Singleton-type bound derived in \cite{sasidharan15}. 

In Section \ref{sec:Sm}, we study the locality of one particular construction of LRCs presented in \cite{silberstein18}. The general idea of this construction is to remove a Simplex code from another Simplex code of higher dimension. It was shown in \cite{silberstein18} that these codes achieve the Griesmer bound and are LRCs with $\delta=2$. The goal in this section is to prove the locality for every dimension and higher $\delta$ and show that this construction leads to optimal LRCs for every locality by the bound \eqref{eq:CMG_r} and to optimal H-LRCs by the new bound derived in Section \ref{sec:bound}.

As a first step, we describe the restrictions of dimension $k-1$ and prove the locality for this dimension using combinatorial techniques. These results allow us to derive the weight enumerator of the constructed codes. As a second step, we use a recursive argument to get all the restrictions of these codes to closed sets. Our main contribution is the complete list of possible localities for these codes. In particular, this shows that the constructed codes are alphabet-optimal H-LRCs. Finally, since a special case of this construction leads to the Reed--Muller codes RM$(1,m)$, we obtain as a corollary to our result that the Reed--Muller codes RM$(1,m)$ are H-LRCs and we derive their locality parameters.

\section{Preliminaries}
\label{sec:def}
%Preliminaries

%Contains : Coding theory notations. 

We denote the set $\{1,2,\ldots, n\}$ by $[n]$ and the set of all subsets of $[n]$ by $2^{[n]}$. The set of all positive integers including $0$ is denoted by $\Z_{+}$. 

The Gaussian coefficient, which counts the number of subspaces of dimension $k$ in the vector space $\F_{q}^{n}$, is denoted by
\[
{n \brack k}_{q}=\prod\limits_{i=0}^{k-1}\frac{q^{n-i}-1}{q^{i+1}-1}.
\]

For a length-$n$ vector $\mathbf{v}$ and a set $I \subseteq [n]$, the vector $\mathbf{v}_{I}$ denotes the restriction of the vector $\mathbf{v}$ to the coordinates in the set $I$. A generator matrix of a linear code $\CC$ is $G_{\CC} = (\mathbf{g_{1}} \cdots \mathbf{g_{n}})$ where $\mathbf{g_{i}} \in \F_{q}^{k}$ is a column vector for $i \in [n]$. The shortening of a code $\CC$ to the set of coordinates $I \subseteq [n]$ is defined by $\CC/I = \{ \mathbf{c}_{[n]\setminus I} : \mathbf{c} \in \CC \text{ such that } c_{i}=0 \text{ for all } i \in I\}$ and the restriction of a code $\CC$ to  $I$ is defined by $\CC|_{I}=\{\mathbf{c}_{I} : \mathbf{c} \in \CC \}$. For convenience, we call the codes obtained by a restriction \emph{restricted codes}. Two linear codes $\CC$ and $\CC'$ are called isomorphic if $\CC'$ can be obtained by a permutation on the coordinates of the codewords of $\CC$.

The support of a codeword $\vec{c} \in \CC$ is $\supp(\vec{c})=\{ i \in [n] : c_{i} \neq 0 \}$ and its weight is $\wt(\vec{c})=|\supp(\vec{c})|$. The weight enumerator of $\CC$ is defined as 
\[
W_{\CC}(x, y)=\sum\limits_{\vec{c} \in \CC}x^{n-\wt(\vec{c})}y^{\wt(\vec{c})}. 
\]

The Simplex code $\Ss(m)$, or sometimes $\Ss_{q}(m)$, is a linear code over $\F_{q}$ obtained via the generator matrix $G_{m}$ consisting of all pairwise linearly independent vectors in $\F_{q}^{m}$. The parameters of $\Ss(m)$ are therefore $[(q^{m}-1)/(q-1), m, q^{m-1}]$.

Since most of the work in this paper is done on subsets of coordinates and restricted codes, it is preferable to use the notion of an entropy function on the subsets $I \subseteq [n]$ that corresponds to the dimension of the restriction to $I$. We only state here the definition for linear codes but it can be generalize to bigger class of codes (see \cite{westerback18}).

\begin{definition}
Let $\CC$ be a linear code of length $n$ and $I \subseteq [n]$. The entropy associated to $\CC$ is the function $H_{\CC} : [n] \to \Z$ with $H_{\CC}(I) = \dim (\CC|_{I})$.
\end{definition}

For ease of notation, if the underlying code of $H_{\CC}$ is clear, we drop the specification to $\CC$. For a subset $I \subseteq [n]$, $H_{\CC}(I)$ is equivalent to the rank of the submatrix of the generator matrix formed by the columns $\mathbf{g}_{i}$ with $i \in I$ or the rank function of $I$ in the associated matroid of $\CC$. As such, it has the following standard properties.

\begin{proposition}
Let $\CC$ be a linear code of length $n$ and $H$ the entropy function associated to $\CC$. For $I,J \subseteq [n]$, we have
\begin{enumerate}
\item $H(I) \leq |I|$,
\item If $I \subseteq J$ then $H(I) \leq H(J)$,
\item $H(I \cup J) + H(I \cap J) \leq H(I) + H(J)$.
\end{enumerate}
\end{proposition}

The entropy function also behaves nicely for restricted codes. Let $I \subseteq [n]$ and $\CC|_{I}$ the restriction of $\CC$ to the set $I$. Then for $J \subseteq I$, we have $H_{\CC|_{I}}(J) = H_{\CC}(J)$.

Finally, we define a closure operation on the subsets of $[n]$ for linear codes. 

\begin{definition}
Let $\CC$ be a linear code of length $n$ and $I \subseteq [n]$. The \emph{closure operator} $\cl : 2^{[n]} \to 2^{[n]}$ is $ \cl(I)=\{ e \in [n] : H(I \cup e) = H(I) \}$. A set $I \subseteq [n]$ is a \emph{closed set} if $\cl(I)=I$. 
\end{definition}

One can think of the closure operator via the generator matrix $G_{\CC}$ of $\CC$ where $\cl(I)$ is the set of all columns in $G_{\CC}$ contained in the linear span of the columns indexed by $I$.

%Matroid part -------------------------------------------------
\subsection{Preliminaries on matroids}

Since we work on Simplex codes which have a lot of combinatorial properties, we use some tools coming from matroid theory. Matroids have many equivalent definitions in the literature. Here, we choose to present matroids via their rank functions. Much of the contents in this part can be found in more detail in \cite{freij18}.

\begin{definition}
\label{def:matroid_rank}
A \emph{(finite) matroid} $M=(E, \rho)$ is a finite set $E$ together with a \emph{rank function} $\rho:2^E \rightarrow \mathbb{Z}$ such that for all subsets $X,Y \subseteq E$,
\[
\begin{array}{rl} 
(R.1) & 0 \leq \rho(X) \leq |X|,\\
(R.2) & X \subseteq Y \quad \Rightarrow \quad \rho(X) \leq \rho(Y),\\
(R.3) & \rho(X) + \rho(Y) \geq \rho(X \cup Y) + \rho(X \cap Y). 
\end{array}
\]
\end{definition}

Any matrix $G$ over a field $\F$ generates a matroid $M_G=(E,\rho)$, where $E$ is the set of columns of $G$, and $\rho(X)$ is the rank of the submatrix of $G$ formed by the columns indexed by $X$. As elementary row operations preserve the row space of $G(X)$ for all $X\subseteq E$, it follows that row-equivalent matrices generate the same matroid.  

Thus, there is a straightforward connection between linear codes and matroids. Let $\CC$ be a linear code over a field $\F$. Then any two different generator matrices of $\CC$ will have the same row space by definition, so they will generate the same matroid. Therefore, without any inconsistency, we can denote the matroid associated to these generator matrices by $M_{\CC} = (E, \rho_{\CC})$ where $\rho_{\CC}(X)$ is the rank of the submatrix of $G_{\CC}$ formed by the columns indexed by $X$.

%Isomorphism of matroids and relation with isomorphism of codes.
Two matroids $M_1 = (E_1, \rho_1)$ and $M_2 = (E_2, \rho_2)$ are \emph{isomorphic} if there exists a bijection $\psi: E_1 \rightarrow E_2$ such that $\rho_2(\psi(X)) = \rho_1(X)$ for all subsets $X \subseteq E_1$. We denote two isomorphic matroids by $M_{1} \approx M_{2}$. This implies in particular that if $\CC_{1}$ and $\CC_{2}$ are two linear codes and $M_{\CC_{1}}$ is isomorphic to $M_{\CC_{2}}$, then $\CC_{1}$ is isomorphic to $\CC_{2}$. 

% Deletion, restriction + link with coding theory. Properties link exchange contract and deletion; union of two deletion (should be obvious..). 
One way of defining a new matroid from an existing one is obtained by restricting the matroid to one of its subsets. For a given set $Y \subseteq E$, we define the \emph{restriction} of $M$ to $Y$ to be the matroid $M|Y = (Y, \rho_{|Y})$ by $\rho_{|Y} (X) = \rho(X)$ for all subsets $X \subseteq Y$. The restriction of $M$ to $E-Y$ is called the \emph{deletion} of $Y$ and is denoted by $M \setminus Y$. The two previous operations correspond to the restriction and puncturing of a linear code $\CC$. 

% Lattice of flat with operators and lessdot. Precise clearly that vee is cl() but wedge is just \cap.
% closure and remark about that it coincides with closure w.r.t, entropy. 
Let $M = (E,\rho)$ be a matroid. The \emph{closure} operator $\cl :2^E \rightarrow 2^E$ is defined by $ \cl(X) = \{e \in E  : \rho(X \cup e) = \rho(X)\}$. A subset $F \subseteq E$ is a \emph{flat} if $\cl(F) = F$. The collection of flats of $M$ is denoted by $\FF(M)$ and forms a lattice under the inclusion. For $F_{1}, F_{2} \in \FF(M)$, the meet is $F_{1} \wedge F_{2} = F_{1} \cap F_{2}$ and join is $F_{1} \vee F_{2} = \cl(F_{1} \cup F_{2})$. We denote by $\lessdot$ the covering relation, \emph{i.e.},  $F_{1} \lessdot F_{2}$ if $F_{1} \subsetneq F_{2}$ and there is no $F_{3} \in \FF(M)$ with $F_{1} \subsetneq F_{3} \subsetneq F_{2}$. A \emph{hyperplane} $H$ is a flat of $M$ with $\rho(H)=\rho(E)-1$ and the collection of all hyperplanes is denoted by $\HH(M)$. One interesting property of $\HH(M)$ is that every flat can be expressed as an intersection of hyperplanes, \ie , $\FF(M)= \left\lbrace \cap_{H \in \YY} H : \YY \subseteq \HH(M) \right\rbrace$. Finally, for $Y \subseteq E$, we can express the flats of $M|Y$ via the flats of $M$ by the relation $\FF(M|Y)=\{ F \cap Y : F \in \FF(M) \}$.

\section{Bound for H-LRCs}
\label{sec:bound}
%input 3
%CM type bound for H-LRC 

Codes with hierarchical locality were introduced in \cite{sasidharan15} to optimize further the number of nodes contacted for repair. In this section, we first focus on H-LRCs with $2$-level hierarchy and derive an alphabet-dependent bound for these codes based on \cite{cadambe15}. Then, we extend this bound to H-LRCs with $h$-level hierarchy.

\begin{definition}
Let $r_{2} \leq r_{1}$ and $\delta_{2} < \delta_{1}$. An $[n,k,d]$ linear code $\CC$ is a code with hierarchical locality having locality parameters $[(r_{1},\delta_{1}),(r_{2},\delta_{2})]$ if for all code symbols $i\in [n]$, there exists a set $M_{i} \subseteq [n]$ such that 
\begin{enumerate}
\item $i \in M_{i}$,
\item $H(M_{i}) \leq r_{1}$,
\item The minimum distance of $\CC|_{M_{i}}$ is at least $\delta_{1}$,
\item $\CC|_{M_{i}}$ is an LRC with $(r_{2},\delta_{2})$-locality.
\end{enumerate}
\end{definition}

The codes $\CC|_{M_{i}}$ are called \textit{middle codes} and their restrictions of dimension $\leq r_2$ and minimum distance $\geq \delta_2$ are called \textit{local codes}. Similarly, the \textit{middle sets} and \textit{local sets} are the sets $M_{i}$ and the sets such that the restrictions to them give the local codes. Notice that, contrary to the standard definition of LRCs, the authors of \cite{sasidharan15} bound the dimension of the restricted codes instead of the size. It was proven in \cite[Theorem 2.1]{sasidharan15} that any $[n,k,d]$ linear code with hierarchical locality $[(r_{1},\delta_{1}),(r_{2},\delta_{2})]$ satisfies the following Singleton-type bound 

\begin{equation}
\label{eq:HLRC_SingletonB}
d \leq n-k+1 - \left\lfloor \frac{k-1}{r_{2}} \right\rfloor (\delta_{2}-1) - \left\lfloor \frac{k-1}{r_{1}} \right\rfloor (\delta_{1}-\delta_{2}).
\end{equation}

To obtain the alphabet-dependent bound in \cite{cadambe15}, the authors proved two results: the construction of a restricted code with small dimension and large size and a lemma about shortened codes. The lemma is the following. 

\begin{lemma}[\hspace{1sp}\cite{cadambe15}, Lemma 2]
\label{lemma:contract}
Let $\CC$ be an $[n,k,d]$ linear code over $\F_{q}$ and $I \subseteq [n]$ such that $H(I) < k $. Then the shortened code $\CC/I$ has parameters $[n-|I|, k-H(I), d'\geq d]$. 
\end{lemma}

%New results---------

Therefore, to get an alphabet-dependent bound for H-LRCs, we need to construct a set with an upper bound on its entropy and such that its size uses the hierarchical locality property to be as large as possible. To achieve this requirement, we modify the construction algorithm used in the proof of the Singleton-type bound in \cite{sasidharan15}. 

\begin{lemma}
\label{lemma:size_construction_h2}
Let $\CC$ be an $[n,k,d]$ H-LRC with locality $[(r_{1},\delta_{1}),(r_{2},\delta_{2})]$ and $\lambda \in \Z_{+}$ with $0 \leq \lambda \leq k$. Then, there exists a set $I_{c}$ such that 
\begin{itemize}
\item $ H(I_{c}) \leq \lambda $,
\item $|I_{c}| \geq \lambda + \left\lfloor \frac{\lambda}{r_{2}}\right\rfloor (\delta_{2}-1) + \left\lfloor \frac{\lambda}{r_{1}} \right\rfloor  (\delta_{1}-\delta_{2})$.
\end{itemize}
\end{lemma}

%Algorithm---------
\begin{algorithm}
\caption{For the proof of Lemma \ref{lemma:size_construction_h2} }\label{algo1}
\begin{algorithmic}[1]
\State Let $i_{1}=0, i_{2}=0, I = \emptyset$.
\While{$\exists$ middle set $M_{j}$ such that $M_{j} \nsubseteq I$ and $H(I \cup M_{j}) \leq \lambda$}
	\While{$\exists$ local set $L_{i} \subset M_{j}$ such that  $L_{i} \nsubseteq I$}
		\State $i_{2}=i_{2}+1$
		\State $a_{i_{2}}=H(I \cup L_{i}) - H(I)$
		\State $s_{i_{2}}=a_{i_{2}} + \delta_{2} -1$
		\State $I=\cl(I \cup L_{i})$
	\EndWhile
	\State $i_{1}=i_{1}+1$
	\State $s_{i_{2}}=s_{i_{2}}-\delta_{2} + \delta_{1}$
\EndWhile
\State Let $M_{x}$ a middle set with $M_{x} \nsubseteq I$. 
\While{$\exists$ local set $L_{i} \subset M_{x}$ such that $L_{i} \nsubseteq I$ and $H(I \cup L_{i}) \leq \lambda$}
\State $i_{2}=i_{2}+1$
\State $a_{i_{2}}=H(I \cup L_{i}) - H(I)$
\State $s_{i_{2}}=a_{i_{2}} + \delta_{2}-1$
\State $I=\cl(I \cup L_{i})$
\EndWhile
\end{algorithmic}
\end{algorithm}
%----------------

\begin{proof}
We use the construction given in Algorithm \ref{algo1} to build the set $I_{c}$. In the algorithm, $a_{i_{2}}$ denotes the incremental entropy and $s_{i_{2}}$ is a lower bound on the incremental size. Indeed if $J \subseteq [n]$ and $L_{i}$ is a local set such that $L_{i} \nsubseteq J $ and $H(J \cup L_{i})>H(J)$, then $|\cl(J\cup L_{i})| - |J| \geq H(J \cup L_{i}) - H(J) + \delta_{2} -1$. Similarly, if $J \subseteq [n]$ and $M_{j}$ is a middle set such that $M_{j} \nsubseteq J $ and $H(J \cup M_{j}) >H(J)$, then $|\cl(J\cup M_{j})| - |J| \geq H(J \cup M_{j}) - H(J) + \delta_{1} -1$. 

Denote by $I$ the set obtained at the end of Algorithm \ref{algo1} and by $i_{1}$ and $j_{2}$ the last indices. The total number of local sets visited in this algorithm is $i_{2} \geq \left\lfloor \frac{\lambda}{r_{2}} \right\rfloor$ since the union of $\left\lfloor \frac{\lambda}{r_{2}} \right\rfloor$ arbitrary local sets has an entropy less than $\lambda$. By the same argument, we have $i_{1} \geq \left\lfloor \frac{\lambda}{r_{1}} \right\rfloor$. 

Let $s_{e}=\lambda - H(I)$ and $A \subseteq [n]$ such that $H(A) = s_{e}$ and $I \cap A= \emptyset$ which is always possible to construct. Finally, let $I_{c}=I \cup A$. We prove now that $I_{c}$ has the desired properties. 

By construction, we have $H(I_{c})=\lambda$. For the size, we have
\begin{align*}
|I_{c}| &= s_{e} + |I| \geq s_{e} + \sum\limits_{i=1}^{i_{2}} s_{i} \\
&= s_{e} + \sum\limits_{i=1}^{i_{2}}(a_{i} + \delta_{2} -1) + i_{1}(\delta_{1}-\delta_{2}) \\
&= \lambda + i_{2}(\delta_{2}-1) + i_{1}(\delta_{1}-\delta_{2}) \\
& \geq \lambda + \left\lfloor \frac{\lambda}{r_{2}} \right\rfloor (\delta_{2} -1) + \left\lfloor \frac{\lambda}{r_{1}} \right\rfloor (\delta_{1} - \delta_{2}).
\end{align*}

Hence, $I_{c}$ has the desired entropy and size. \qed
\end{proof}

Thus, we obtain the following bound for 2-level H-LRCs. 

\begin{theorem}
\label{thm:CM_HLRC}
Let $\CC$ be an $[n,k,d]$ H-LRC over $\F_{q}$ with locality $[(r_{1},\delta_{1}),(r_{2},\delta_{2})]$. Then we have
\begin{equation}
\label{eq:CM_HLRC}
k \leq \min_{\lambda \in \Z_{+}} \left\{ \lambda + k_{\opt}^{(q)}(n-\nu, d ) \right\}
\end{equation}
where $\nu = \lambda + \left\lfloor \frac{\lambda}{r_{2}}\right\rfloor (\delta_{2}-1) + \left\lfloor \frac{\lambda}{r_{1}} \right\rfloor  (\delta_{1}-\delta_{2})$ and $k_{\opt}^{(q)}(n,d)$ is the largest possible dimension of a code of length $n$, for a given alphabet size $q$ and minimum distance $d$. 
\end{theorem}

\begin{proof}
The theorem follows from Lemma \ref{lemma:contract} and \ref{lemma:size_construction_h2} in the same manner as in \cite{cadambe15} by using the estimation of the size coming from Lemma \ref{lemma:size_construction_h2}. 
\qed
\end{proof}

Lemma \ref{lemma:size_construction_h2} can be seen as a proof of concept that we can modify the algorithms in \cite{sasidharan15} to obtain an alphabet-dependent bound. Indeed, the algorithm presented here and the one presented in \cite{sasidharan15} are equivalent in the sense that if $\lambda=k-1$ and $I_{c}$ is the set obtained by the algorithm in the proof of Lemma \ref{lemma:size_construction_h2}, we obtain again the Singleton-type bound \eqref{eq:HLRC_SingletonB} via the relation $d \leq n- |I_{c}|$. This implies that the bound \eqref{eq:CM_HLRC} is at least as good as the Singleton-type bound \eqref{eq:HLRC_SingletonB}.

\subsection{Extension of the bound to arbitrary levels of hierarchy }

In this part, we extend the alphabet-dependent bound to H-LRCs with $h$-level hierarchical locality. The definition and notation of these codes is the following. 

\begin{definition}
An $[n,k,d]$ linear code $\CC$ is a code with \emph{$h$-level hierarchical locality} having locality parameters $[(r_{1}, \delta_{1}), (r_{2}, \delta_{2}), \ldots , (r_{h}, \delta_{h})]$ if for all code symbols $i \in [n]$, there exists a collection of sets $\{L_{1_{i}}, L_{2_{i}}, \ldots , L_{h_{i}} \}$ such that for every $j \in [h]$, we have
\begin{enumerate}
\item $i \in L_{j_{i}} \subseteq [n]$,
\item $H(L_{j_{i}}) \leq r_{j}$
\item The minimum distance of $\CC|_{L_{j_{i}}}$ is at least $\delta_{j}$,
\item $\CC|_{L_{j_{i}}}$ is a code with $(h-j)$-level hierarchical locality having locality 

parameters $[(r_{j+1}, \delta_{j+1}), (r_{j+2}, \delta_{j+2}), \ldots , (r_{h}, \delta_{h})]$.
\end{enumerate}
\end{definition}

A set $L_{j_{i}}$ for $j \in [h]$ is referred to as a \emph{level-$j$} set. It was proven in \cite[Theorem 3.1]{sasidharan15} that any $[n,k,d]$ H-LRC with hierarchical locality $[(r_{1}, \delta_{1}), \ldots , (r_{h}, \delta_{h})]$ satisfies the following Singleton-type bound 

\begin{equation}
\label{eq:HLRC_SingletonB_h}
d \leq n-k+1 - \left\lfloor \frac{k-1}{r_{h}} \right\rfloor (\delta_{h}-1) - \sum\limits_{l=1}^{h-1} \left\lfloor \frac{k-1}{r_{l}} \right\rfloor (\delta_{l}-\delta_{l+1}).
\end{equation}

Obtaining an alphabet-dependent bound for these codes is done in the same manner as in the case of $h=2$. 

\begin{lemma}
\label{lemma:size_construction_h}
Let $\CC$ be an $[n,k,d]$ H-LRC with locality $[(r_{1}, \delta_{1}), \ldots , (r_{h}, \delta_{h})]$ and $\lambda \in \Z$ with $0 \leq \lambda \leq k$. Then, there exists a set $I_{c}$ such that 
\begin{itemize}
\item $H(I_{c}) \leq \lambda$,
\item $|I_{c}| \geq \lambda + \left\lfloor \frac{\lambda}{r_{h}} \right\rfloor (\delta_{h}-1) + \sum\limits_{l=1}^{h-1} \left\lfloor \frac{\lambda}{r_{l}} \right\rfloor (\delta_{l}-\delta_{l+1}) $.
\end{itemize}
\end{lemma}

%Algorithm--------------------------------------
\begin{algorithm}
\caption{For the proof of Lemma \ref{lemma:size_construction_h} }\label{algo2}
\begin{algorithmic}[1]
\State Let $i_{1}=0, i_{2}=0, \ldots , i_{h}=0$; $I=\emptyset$;
\State $M_{0}=[n], M_{1}=\emptyset, M_{2}=\emptyset, \ldots , M_{h}=\emptyset$.
\For{$j\gets 1, h$}
	\While{$\exists$ set $L_{j_{i}}$ such that $L_{j_{i}} \nsubseteq I$ and $H(I \cup L_{j_{i}}) \leq \lambda$}
		\State $M_{j-1}=L_{j_{i}}, M_{j}=L_{j_{i}}, l=j$.
		\While{$\cl(I \cup L_{j_{i}}) \neq I$}
			\If{$\exists$ a level-$l$ set $L_{l_{i}} \subseteq M_{l-1}$ such that $L_{l_{i}} \nsubseteq I$}
				\State $M_{l}=L_{l_{i}}$.
				\If{$l$ equals $h$}
					\State $i_{h}=i_{h}+1$, $a_{i_{h}}=H(I \cup M_{h}) - H(I)$, $s_{i_{h}}=a_{i_{h}} + \delta_{h} -1$.
					\State $I=\cl(I \cup M_{h})$
				\Else
					\State $l=l+1$
				\EndIf
			\Else
				\State $l=l-1$
				\State $i_{l}=i_{l}+1$, $s_{i_{h}}=s_{i_{h}} + \delta_{l} - \delta_{l+1}$.
			\EndIf
		\EndWhile
		\For{$l \gets j, h-1$}
			\State $i_{l}=i_{l}+1$
		\EndFor
		\State $s_{i_{h}}=s_{i_{h}} + \delta_{j} - \delta_{h}$
	\EndWhile
\EndFor
\end{algorithmic}
\end{algorithm}
%---------------------------------------------

\begin{proof}
We use the construction given in Algorithm \ref{algo2} to build the set $I_{c}$. The basic idea is the same as in the proof of Lemma \ref{lemma:size_construction_h2}. First, the algorithm identifies the smallest $j$ such that a level-$j$ set $L_{j_{i}}$ has $H(L_{j_{i}}) \leq \lambda$. This is important when $\lambda < r_{1}$ or during the process of the algorithm when $\lambda-H(I)$ becomes small. After this set is found, the algorithm visits recursively the sets $L_{l_{i}}$ with $l \in [j,h-1]$ such that each of them is contained in the previous one and identifies a level-$(h-1)$ set $L_{(h-1)_{i}}$. Then, it starts adding to $I$ the level-$h$ sets contained in $L_{(h-1)_{i}}$. When every symbol of $L_{(h-1)_{i}}$ has been added to $I$, it steps back one level, finds a new $L_{(h-1)_{i'}}$ that contains new symbols not in $I$, and adds them in the same manner using level-$h$ sets. At some point, the algorithm adds to $I$ a level-$h$ set $L_{h_{i''}}$ containing the last remaining symbols of $L_{j_{i}}$. The second \textbf{while} loop is not satisfied anymore and because the last added set was a level-$h$ set, the algorithm adds one to the count of the visited sets per level between $j$ and $h-1$. Next, it pursues by identifying another level-$j$ set $L_{j_{i'}}$ satisfying $H(I \cup L_{j_{i'}}) \leq \lambda$ or, if all level-$j$ sets exceed $\lambda$, the algorithm searches for valid level-$(j+1)$ sets.

As in the proof of Lemma \ref{lemma:size_construction_h2}, $a_{i_{h}}$ denotes the incremental entropy and $s_{i_{h}}$ is a lower bound on the incremental size corrected accordingly to the level of the added set. The counters $i_{j}$ for $j \in [h]$ are the number of level-$j$ sets visited by the algorithm. 

When the algorithm terminates, the counters are lower-bounded by 
\[
i_{l} \geq \left\lfloor \frac{\lambda}{r_{l}} \right\rfloor, \; 1 \leq l \leq h
\]
since the union of $\left\lfloor \frac{\lambda}{r_{l}} \right\rfloor$ arbitrary level-$l$ sets $L_{j_{l}}$ has an entropy less than $\lambda$. 

Denote by $I$ the set obtained at the end of the algorithm and by $i_{1}, \ldots ,  i_{h}$ the last values of the counters. Let $s_{e}=\lambda-H(I)$ and $A \subseteq [n]$ such that $H(A)=s_{e}$ and $I \cap A =\emptyset$. Finally, let $I_{c}=I \cup A$. We prove now that $I_{c}$ has the desired properties. By construction, we have $H(I_{c})=\lambda$. For the size, we have
\begin{align*}
|I_{c}| & = s_{e} + |I| \geq s_{e} + \sum\limits_{i=1}^{i_{h}} s_{i} \\
& = s_{e} + \sum\limits_{i=1}^{i_{h}}a_{i} + \sum\limits_{i=1}^{i_{h}}(\delta_{h}-1) + i_{h-1}(\delta_{h-1}-\delta_{h}) + i_{h-2}(\delta_{h-2}-\delta_{h-1}) + \ldots + i_{1}(\delta_{1}-\delta_{2}) \\
& = \lambda + i_{h}(\delta_{h}-1) + \sum\limits_{l=1}^{h-1}i_{l}(\delta_{l}-\delta_{l+1}) \\
& \geq \lambda + \left\lfloor \frac{\lambda}{r_{h}} \right\rfloor (\delta_{h}-1) + \sum\limits_{l=1}^{h-1} \left\lfloor \frac{\lambda}{r_{l}} \right\rfloor (\delta_{l}-\delta_{l+1}).
\end{align*}
Hence $I_{c}$ has the desired entropy and size. 
\end{proof}

Thus, we obtain the following extension of the bound \eqref{eq:CM_HLRC} for $h$-level H-LRCs. 

\begin{theorem}
\label{thm:CM_HLRC_h}
Let $\CC$ be an $[n,k,d]$ H-LRC over $\F_{q}$ with locality $[(r_{1}, \delta_{1}), \ldots , (r_{h}, \delta_{h})]$. Then, we have
\begin{equation}
\label{eq:CM_HLRC_h}
k \leq \min_{\lambda \in \Z_{+}} \left\{ \lambda + k_{\opt}^{(q)}(n-\nu, d ) \right\}
\end{equation}
where $\nu = \lambda + \left\lfloor \frac{\lambda}{r_{h}} \right\rfloor (\delta_{h}-1) + \sum\limits_{l=1}^{h-1} \left\lfloor \frac{\lambda}{r_{l}} \right\rfloor (\delta_{l}-\delta_{l+1})$ and $k_{\opt}^{(q)}(n,d)$ is the largest possible dimension of a code of length $n$, for a given alphabet size $q$ and minimum distance $d$.
\end{theorem}

As in the case $h=2$, for $\lambda=k-1$ and $I_{c}$ the set obtained by Algorithm \ref{algo2}, we obtain again the Singleton-type bound \eqref{eq:HLRC_SingletonB_h} via the relation $d \leq n-|I_{c}|$, which shows that the bound \eqref{eq:CM_HLRC_h} is at least as good as the Singleton-type bound \eqref{eq:HLRC_SingletonB_h}. Moreover, the bound \eqref{eq:CM_HLRC_h} yields that H-LRCs achieving any bound on the parameters $[n,k,d]$ only are directly alphabet-optimal H-LRCs by setting $\lambda=0$.

\section{Hierarchical locality of $\Ss(m)-\Ss(s)$}
\label{sec:Sm}
%Input 4
%Study of locality and H-locality of S(m)-S(s)

%Intro

In \cite{silberstein18}, the authors presented four different constructions of linear LRCs with small locality and high availability. The constructions are based on a method developed in \cite{farrell70} where the generator matrix of a code is obtained by deleting specific columns from the generator matrix of a Simplex code. In this section, we are interested in the locality of one particular construction in \cite{silberstein18} where the deleted columns form again a Simplex code. This construction is highly combinatorial and yields optimal codes achieving the Griesmer bound. The objective is to describe the locality parameters for $\delta>2$ and all dimensions. We show that these codes are locally repairable codes for every dimension implying a complete optimization of the number of nodes contacted for repair according to the number of failures. Moreover, using combinatorial techniques, we establish the complete list of possible locality and show how the local sets can be arranged to form a hierarchical locality. Finally, we prove that these codes are optimal LRCs for all localities and alphabet-optimal H-LRCs by the new bound \eqref{eq:CM_HLRC_h}. 

%Natalia's construction and results

We start by formally defining the construction of these linear LRCs.  

\begin{construction}[\hspace{1sp}\cite{silberstein18} Construction \textrm{IV}]
\label{cst:Sm_Ss}
%hspace to remove the extra space from \cite command. 
Let $G_{m}$ be an $m \times \frac{q^{m}-1}{q-1}$ generator matrix of the Simplex code $\Ss(m)$ and $G_{s}$ an $s \times \frac{q^{s}-1}{q-1}$ generator matrix of the Simplex code $\Ss(s)$ with $s\leq m$. Let $G'_{s}$ be the generator matrix obtained by prepending $m-s$ zeros to every column of $G_{s}$. Let $G_{\CC}$ be the $m \times \frac{q^{m}-q^{s}}{q-1}$ matrix obtained by deleting the $\frac{q^{s}-1}{q-1}$ columns of $G'_{s}$ from $G_{m}$. Then $G_{\CC}$ generates a linear code $\CC$ over $\F_{q}$ denoted by $\Ss(m)-\Ss(s)$. 
\end{construction}

It was proven in \cite[Theorem 14]{silberstein18} that the code $\Ss(m)-\Ss(s)$ with $m \geq 3$ and $s \in [2,m-1]$ is a $[\frac{q^{m}-q^{s}}{q-1}, m, q^{m-1}-q^{s-1}]$ linear LRC over $\F_{q}$ with locality $(r=2,\delta=2)$ if $q>2$ or if $q=2$ and $s<m-1$, and with locality $(r=3,\delta=2)$ when $q=2$ and $s=m-1$. Moreover, $\Ss(m)-\Ss(s)$ achieves the Griesmer bound by \cite[Lemma 16]{silberstein18}. Notice that the code $\Ss(m)-\Ss(m-1)$ is isomorphic to the Reed--Muller code RM$(1,m-1)$. 

The following example illustrates Construction \ref{cst:Sm_Ss}. 

\begin{example}
\label{ex:1246}
Let $G_{4}$ and $G_{2}$ be the generator matrices of the binary Simplex codes $\Ss_{2}(4)$ and $\Ss_{2}(2)$ respectively. Then $\CC = \Ss_{2}(4) - \Ss_{2}(2)$ is a binary $[12,4,6]$ code generated by the matrix
\begin{center}
\begin{tikzpicture}[
    every left delimiter/.style={xshift=.4em},
    every right delimiter/.style={xshift=-.4em},
    ]
    \matrix[
        matrix of math nodes,
        row sep=.01ex,
        column sep=.01ex,
        inner sep=2.5pt,
        left delimiter=(,right delimiter=)
        %nodes={text width=.3em, text height=0.75ex, text depth=.2ex, align=center}
        ] (m) 
        {
		1&0&0&0&1&1&1&0&0&0&1&1&1&0&1\\
		0&1&0&0&1&0&0&1&1&0&1&1&0&1&1\\
		0&0&1&0&0&1&0&1&0&1&1&0&1&1&1\\
		0&0&0&1&0&0&1&0&1&1&0&1&1&1&1\\
        };
        \begin{scope}[on background layer]
            \node[inner sep=0.5pt, fit=(m-1-3)(m-4-3), draw=gray, fill=gray, rounded corners] {};
            \node[inner sep=0.5pt, fit=(m-1-4)(m-4-4), draw=gray, fill=gray, rounded corners] {};
            \node[inner sep=0.5pt, fit=(m-1-10)(m-4-10), draw=gray, fill=gray, rounded corners] {};
        \end{scope} 
\end{tikzpicture}
\end{center}
where the shadowed columns $3,4$, and $10$ are deleted. 

\end{example}

\subsection{Locality of $\Ss(m)-\Ss(s)$ with dimension $m-1$}

The goal of this subsection is to obtain the locality of $\Ss(m)-\Ss(s)$ with a dimension of $m-1$. For this, we make a detour to matroid theory by studying the relation between the hyperplanes of the matroid associated to the Simplex code $\Ss(m)$ and the hyperplanes of the matroid associated to $\Ss(m)-\Ss(s)$. Indeed, the Simplex code has intrinsically a lot of useful combinatorial structures and Construction \ref{cst:Sm_Ss} corresponds to a deletion in matroid theory. Therefore, matroid theory is used here as a tool to understand the closed sets of $\Ss(m)-\Ss(s)$ of dimension $m-1$ and to construct the local set for every code symbol. We start by presenting a lemma that gives the relation between a flat and the hyperplanes of the matroid associated to the Simplex code $\Ss(m)$. 

\begin{lemma}
\label{lemma:coatom_Y}
Let $M_{\Ss}$ be the matroid associated to the Simplex code $\Ss(m)$ and $Y \in \FF(M_{\Ss})$ a flat with $\rho(Y)\leq m-1$. Then, for all hyperplanes $H \in \HH(M_{\Ss})$ either $Y \subseteq H$ or $H \cap Y \lessdot Y$.
\end{lemma}

\begin{proof}
Suppose $Y \nsubseteq H$. Then $H \vee Y = E$. Now $\FF(M_{\Ss})$ is isomorphic to the lattice of linear spaces of $\F_{q}^{m}$. Therefore $\FF(M_{\Ss})$ is a modular lattice and the intervals $[H, H \vee Y]_{\FF(M_{\Ss})}$ and $[H \cap Y, Y]_{\FF(M_{\Ss})}$ are two isomorphic sublattices. Since $[H,H \vee Y]_{\FF(M_{\Ss})} = \{ H, E \}$, then $[H \cap Y, Y]_{\FF(M_{\Ss})} = \{ H \cap Y, Y \}$ and $H \cap Y \lessdot Y$. 
\qed
\end{proof}

Because the extended $\Ss(s)$ in Construction \ref{cst:Sm_Ss} is a closed set, we use the previous lemma applied to the extended $\Ss(s)$ to understand the hyperplanes of $\Ss(m)-\Ss(s)$ via the hyperplanes of $\Ss(m)$ and the deletion corresponding to removing the columns in the generator matrix of $\Ss(m)$. 

\begin{proposition}
\label{prop:surj_map}
Let $M_{\Ss}=(E_{\Ss}, \rho_{\Ss})$ be the matroid associated to the Simplex code $\Ss(m)$ and $M_{\CC}=(E_{\CC}, \rho_{\CC})$ the matroid associated to the code $\CC = \Ss(m)-\Ss(s)$. Let also $Y \subseteq E_{\Ss}$ be such that $M_{\Ss} \setminus Y = M_{\CC}$. Then, the map 
\begin{align*}
\phi : \{ H \in \HH(M_{\Ss}) \text{ with } Y \neq H\}
& \to \HH(M_{\CC}) \\
 H & \mapsto H - H \cap Y
\end{align*}
is a bijection. 
\end{proposition}

\begin{proof}
By construction, the image of $\phi$ is a subset of $2^{E_{\CC}}$. Firstly, we prove that $\HH(M_{\CC})$ is contained in the image of $\phi$. Secondly, we prove that $\phi$ is well-defined which implies with the first part that $\phi$ is a surjection. Finally, we prove that $\phi$ is an injection. 

%This implies that the map produces flats and that {coatom of M_{\CC}} \subseteq image of phi. 
The flats of $M_{\CC}$ can be obtained by the flats of $M_{\Ss}$ via the relation
\[
\FF(M|(E-Y)) = \{ F \cap (E-Y) : F \in \FF(M)  \}
\]
and $F \cap (E-Y) = F-Y = F-F \cap Y$. 

Therefore, we have that $\phi(H) \in \FF(M_{\CC})$. Moreover, since $\rho_{\CC}(E_{\CC})=m$, all hyperplanes of $M_{\CC}$ have a rank equal to $m-1$. Combining this with the fact that $\rho_{\CC}(F-F\cap Y) = \rho_{\Ss}(F-F\cap Y) \leq \rho_{\Ss}(F)$ implies that if $H_{\CC} \in \HH(M_{\CC})$, then there exists $H \in \HH(M_{\Ss})$ such that $H_{\CC} = H - H \cap Y$. Thus $\HH(M_{\CC})$ is contained in the image of $\phi$. 

%The second part consists of proving that rho(phi(Z^{c}))=m-1 proving that the map phi is well defined (no other flats are produces by phi than coatom of M_{\CC}. 

Let $H \in \HH(M_{\Ss})$ and let $H_{\CC} = \phi(H)$ a set. To show that $\phi$ is well-defined, it is enough to prove that $\rho_{\CC}(\phi(H))=m-1$, since $\phi(H) \in \FF(M_{\CC})$. For this, we use an argument from coding theory. Since the restricted code $\Ss(m)|_{H}$ is isomorphic to the Simplex code $\Ss(m-1)$, its minimum distance is equal to $q^{m-2}$. By construction of $\CC$ and $M_{\CC}$, the size of $Y$ is $|Y| = \frac{q^{s}-1}{q-1}$. We distinguish two cases.

Assume that $s<m-1$. Let $M_{H} = M_{\Ss}|H$. By restriction and deletion definitions, we have 
\[
\rho_{\CC}(H_{\CC}) = \rho_{\Ss}(H_{\CC}) = \rho_{M_{H}}(H_{\CC}).
\]
Let $d_{M_{H}} = q^{m-2}$ be the minimum distance of $M_{H}$. We have
\[
d_{M_{H}} = \min \{ |A| : A \subseteq H \text{ and } \rho_{M_{H}}(H - A) < \rho_{M_{H}}(H) \}.
\]
Since $|H \cap Y| \leq |Y| =\frac{q^{s}-1}{q-1} < q^{m-2}$, we have $\rho_{M_{H}}(H - H \cap Y) = \rho_{M_{H}}(H)=m-1$. Thus, $\rho_{\CC}(H_{\CC}) = m-1$. 

Assume now that $s=m-1$. Then, there is a unique hyperplane of $M_{\Ss}$ that contains $Y$ which is $Y$ itself. Let $H \in \HH(M_{\Ss})$ with $H \neq Y$. By Lemma \ref{lemma:coatom_Y}, we have $H \cap Y \lessdot Y$. Since the restricted code $\Ss(m)|_{H \cap Y}$ is isomorphic to the Simplex code $\Ss(s-1)$, we have $|H \cap Y| = \frac{q^{s-1}-1}{q-1}$. By using the same argument on the minimum distance of $M_{H}$, we have 
\[
\rho_{\CC}(H - H \cap Y) = \rho_{M_{H}}(H) = m-1.
\]

Hence, $\phi$ is well-defined and since the image of $\phi$ contains $\HH(M_{\CC})$, $\phi$ is a surjection. 

To prove the injection, suppose that $\phi(H_{1})=\phi(H_{2})$ for $H_{1}, H_{2} \in \HH(M_{\Ss})$ with $Y \neq H_{1}, H_{2}$. Then, $H_{1} - H_{1} \cap Y = H_{2} - H_{2} \cap Y$. Now we look at these sets in $M_{\Ss}$. Since we showed that $\rho_{\Ss}(H_{1} - H_{1} \cap Y) = m-1$, we have that $\cl(H_{1} - H_{1} \cap Y) = H_{1}$ because $H_{1}$ is a hyperplane containing $H_{1} - H_{1} \cap Y$. Similarly, we have $\cl(H_{2} - H_{2} \cap Y) = H_{2}$. Hence $H_{1} = H_{2}$ and $\phi$ is an injection. 
\qed
\end{proof}

The map in Proposition \ref{prop:surj_map} gives us the relation between the hyperplanes of $M_{\CC}$ and the hyperplanes of the matroid $M_{\Ss}$. We can now completely describe the restrictions of $\CC$ to hyperplanes and see that these restrictions are in fact isomorphic to certain codes obtained by Construction \ref{cst:Sm_Ss}. 

\begin{proposition}
\label{prop:coatom_isom}
Let $\CC$ be the code $\CC= \Ss(m)-\Ss(s)$, $M_{\CC}$ the matroid associated to $\CC$, and $H_{\CC} \in \HH(M_{\CC})$. Then $\CC|_{H_{\CC}}$ is either isomorphic to the code $\Ss(m-1)-\Ss(s-1)$ or to $\Ss(m-1)-\Ss(s)$.
\end{proposition}

\begin{proof}
Let $M_{\Ss}$ be the matroid associated with $\Ss(m)$ and $Y\in \FF(M_{\Ss})$ such that $M_{\CC}=M_{\Ss} \setminus~Y$. By Proposition \ref{prop:surj_map}, there exists $H \in \HH(M_{\Ss})$ such that $H_{\CC} = H - H \cap Y$. We prove the isomorphism by distinguishing two cases depending on Lemma \ref{lemma:coatom_Y}. 

Suppose that $Y \subseteq H$. By restriction properties, we have 
\[
M_{\CC}| H_{\CC} = M_{\Ss} \setminus Y | H_{\CC} = M_{\Ss} \setminus Y | (H -Y) = M_{\Ss}|H \setminus Y.
\] 

By construction of the Simplex code, the matroid $M_{\Ss}|H $ is isomorphic to $M(\Ss(m-1))$, the matroid associated to the Simplex code $\Ss(m-1)$. Since $Y \in \FF(M_{\Ss}|H)$, we have that $(M_{\Ss}|H) | Y = M_{\Ss} | Y$ is isomorphic to $M(\Ss(s))$. 

Let $E_{H}$ and $E_{\Ss(m-1)}$ be the ground sets of $M_{\Ss}|H $ and $M(\Ss(m-1))$ respectively. Then, there exists a bijection $\varphi: E_{H} \to E_{\Ss(m-1)}$ that preserves the rank. Let $G_{\Ss(m-1)}$ be a matrix associated to $M(\Ss(m-1))$, $\{ \mathbf{g}_{i} : i \in E_{\Ss(m-1)} \}$ the set of columns of $G_{\Ss(m-1)}$ and $\hat{Y}= \{ \varphi(e) : e \in Y \} $. Then $M_{\Ss}|H \setminus Y$ is isomorphic to $M(G_{\Ss(m-1)} - \{ \mathbf{g}_{i} : i \in \hat{Y} \})$. Since the submatrix formed by the columns indexed in $\hat{Y}$ has a rank equal to $s$, we can perform some suitable row operations on $G_{\Ss(m-1)}$ to transform the columns $\mathbf{g}_{i}$ with $i \in \hat{Y}$ such that they are of the form $\tilde{\mathbf{g}}_{i}=(0,0,\ldots, 0, \tilde{g}_{i_{1}}, \ldots, \tilde{g}_{i_{s}}) \in \F_{q}^{m-1}$. Let $\tilde{G}_{\Ss(m-1)}$ be the matrix obtained after the row operations. Then the matrix $\tilde{G}_{\Ss(m-1)} - \{ \tilde{\mathbf{g}}_{i} : i \in \hat{Y} \}$ is exactly the generator matrix of the code $\Ss(m-1) - \Ss(s)$. By the previous isomorphism, we have $M(\tilde{G}_{\Ss(m-1)} - \{ \tilde{\mathbf{g}}_{i} : i \in \hat{Y} \}) \approx M(G_{\Ss(m-1)} - \{ \mathbf{g}_{i} : i \in \hat{Y} \}) \approx M_{\Ss}|H \setminus Y \approx M_{\CC}|H_{\CC}$. Hence we have indeed that $C|_{H_{\CC}}$ is isomorphic to  $\Ss(m-1)-\Ss(s)$. 

For the other case, assume now that  $Y \nsubseteq H$. By Lemma \ref{lemma:coatom_Y}, we have that $H_{\CC}=H - H \cap Y$ with $H \cap Y \subsetneq Y$. Now this case follows the previous case in a similar manner by replacing $Y$ by $H \cap Y$ and $s$ by $s-1$ since $M_{\Ss}|(H \cap Y)$ is isomorphic to $\Ss(s-1)$. Therefore the same type of isomorphisms yields that $C|_{H_{\CC}}$ is isomorphic to  $\Ss(m-1)-\Ss(s-1)$. 
\qed
\end{proof}

It remains to show the existence of such closed sets for every code symbol in order to prove that the code $\Ss(m)-\Ss(s)$ is an LRC with locality obtained by restrictions to closed sets of dimension $m-1$.

\begin{theorem}
\label{thm:local_isom}
Let $\CC$ be the linear code $\CC=\Ss(m)-\Ss(s)$ of length $n$ with $m \geq 2$ and $s\geq 0$. Then, for all $e \in [n]$ we have the following. 
\begin{itemize}
\item If $s=0$, there is a set $H \subseteq [n]$ containing $e$ such that $\CC|_{H}$ is isomorphic to $\Ss(m-1)$. 
\item If $s=m-1$, there is a set $H \subseteq [n]$ containing $e$ such that $\CC|_{H}$ is isomorphic to $\Ss(m-1)-\Ss(m-2)$. 
\item If $1 \leq s \leq m-2$, there exist two sets $H_{1}, H_{2}$ containing $e$ with $H_{1} \nsubseteq H_{2}$ such that $\CC|_{H_{1}}$ is isomorphic to $\Ss(m-1)-\Ss(s)$ and $\CC|_{H_{2}}$ is isomorphic to $\Ss(m-1)-\Ss(s-1)$. 
\end{itemize}
\end{theorem}

\begin{proof}
Let $M_{\Ss}=(E_{\Ss}, \rho_{\Ss})$ be the matroid associated with the Simplex code $\Ss(m)$ and $Y \subseteq E_{\Ss}$ such that $M_{\CC}=M_{\Ss} \setminus Y$. Let $e \in E_{\Ss}\setminus Y$. The general idea of the proof is the following. First, we construct a specific hyperplane in $\HH(M_{\Ss})$ containing $e$. Secondly, we use Proposition \ref{prop:surj_map} to get a hyperplane of $M_{\CC}$. Finally, we apply Proposition \ref{prop:coatom_isom} to this hyperplane to obtain the isomorphism. 

If $s=0$, then the code $\CC$ is the Simplex code $\Ss(m)$ and there is therefore a hyperplane $H \in \HH(M_{\CC})$ containing $e$ such that $\CC|_{H}$ is isomorphic to $\Ss(m-1)$. 

%s=m-1
If $s=m-1$, then $Y \in \HH(M_{\Ss})$ and is the only hyperplane containing $Y$. We know that there exists at least one hyperplane in $\HH(M_{\Ss})$ that contains $e$. Let $H_{\Ss}$ be such a hyperplane. By Proposition \ref{prop:surj_map}, $H:=H_{\Ss} - H_{\Ss} \cap Y$ is a hyperplane of $M_{\CC}$ containing $e$. Applying Proposition \ref{prop:coatom_isom} yields that $\CC|_{H}$ is isomorphic to $\Ss(m-1) - \Ss(s-1)$ since $Y$ is not contained in $H_{\Ss}$.

%s<m-1
Suppose now that $1 \leq s\leq m-2$. Since $\rho_{\Ss}(Y \cup \{e \} ) \leq m-1$, there is a hyperplane $H_{Y,e}$ in $\HH(M_{\Ss})$ that contains $Y \cup e$. Therefore, Proposition \ref{prop:surj_map} and Proposition \ref{prop:coatom_isom} yield that $H_{1}:=H_{Y,e} - Y$ is a hyperplane of $M_{\CC}$ containing $e$ and $\CC|_{H_{1}}$ is isomorphic to $\Ss(m-1)-\Ss(s)$. 

Let $X \in \FF(M_{\Ss})$ such that $X \lessdot Y$ and let $X_{e} = \cl(X \cup \{ e \} ) \in \FF(M_{\Ss})$. By the hyperplanes property, we have $X_{e} = \bigcap \{ H_{\Ss} : X_{e} \subset H_{\Ss} \text{ and } H_{\Ss} \in \HH(M_{\Ss} ) \}$. Thus, there exists $H_{X,e} \in \HH(M_{\Ss})$ such that $X_{e} \subseteq H_{X,e}$ and $Y \nsubseteq H_{X,e}$. Applying Proposition \ref{prop:surj_map} and Proposition \ref{prop:coatom_isom} yields that $ H_{2}:=H_{X,e} -  H_{X,e} \cap Y =  H_{X,e} - X$ is a hyperplane of $M_{\CC}$ containing $e$ and $\CC|_{H_{2}}$ is isomorphic to $\Ss(m-1)-\Ss(s-1)$. 
\qed
\end{proof}

%Thm 4 gives existence of coatom, Proposition gives the "uniqueness". Therefore, no other closed set of rank m-1 with different size and, or min dist. 

Theorem \ref{thm:local_isom} can be seen as showing the existence of certain hyperplanes while Proposition \ref{prop:coatom_isom} is of the form of a uniqueness statement on the parameters size, dimension and minimum distance of the hyperplanes. Therefore, the two combined with Proposition \ref{prop:surj_map} yield the complete characterization of all the hyperplanes of the matroid associated to $\Ss(m)-\Ss(s)$ and thus the characterization of all restrictions of $\Ss(m)-\Ss(s)$ to closed sets of dimension $m-1$.

\subsection{Weight enumerator of $\Ss(m)-\Ss(s)$}

Before we continue deriving the localities with dimension less than $m-1$, the results developed so far allow us to compute the weight enumerator of the codes obtained by Construction \ref{cst:Sm_Ss}. For this, we use a theorem from \cite{oxley92} that links the hyperplanes and the codeword supports. 

\begin{theorem}[\hspace{1sp}\cite{oxley92} Theorem 9.2.4]
\label{thm:hyperplanes_cdwd_support}
For each linear code $\CC$, the hyperplanes of $M_{\CC}$ are precisely the complement of the minimal non-empty codeword supports of $\CC$.
\end{theorem}

In order to compute the weight enumerator of $\CC=\Ss(m)-\Ss(s)$, the idea is to associate the codewords of $\CC$ with the codewords of $\Ss(m)$ and the hyperplanes of $M_{\Ss}$. Then, we can use Lemma \ref{lemma:coatom_Y} to understand the effect of the puncturing of $\Ss(m)$ on the hyperplanes of $M_{\Ss}$. We start by a lemma that links the codewords of $\Ss(m)$ and the hyperplanes of $M_{\Ss}$. 

\begin{lemma}
\label{lemma:map_class_cdwd_hyperplanes}
Let $\Ss_{q}(m)$ be the Simplex code of dimension $m$ over $\F_{q}$ and $M_{\Ss}$ the associated matroid. Define $\sim$ as the equivalence relation on the non-zero codewords of $\Ss(m)$ given by $\vc \sim \vc' $ if $\vc = a \vc'$ with $a \in \F_{q}^{\ast}$. Then the map
\begin{align*}
\varphi : \{ \vc : \vc \in \Ss(m), \vc \neq \mathbf{0} \} / \sim & \to \HH(M_{\Ss}) \\
 [\vc] & \mapsto [n]-\supp(\vc)
\end{align*}
is a bijection. 
\end{lemma}

\begin{proof}
It is clear that $\sim$ is indeed an equivalence relation. Furthermore, the map is well-defined since all multiples of a codewords share the same support. Then, Theorem \ref{thm:hyperplanes_cdwd_support} implies that $\varphi$ is a surjection. Now, there are $\frac{q^{m}-1}{q-1}$ equivalence classes of codewords and ${m \brack m-1}_{q}=\frac{q^{m}-1}{q-1}$ hyperplanes in $M_{\Ss}$ since they are exactly the linear spaces of dimension $m-1$ in $\F_{q}^{m}$. Hence $\varphi$ is a bijection. 
\qed
\end{proof}

We can now state the formula for the weight enumerator of $\CC$. 

\begin{theorem}
The weight enumerator of the code $\CC=\Ss_{q}(m)-\Ss_{q}(s)$ over $\F_{q}$ is
\[
W_{\CC}(x, y) = x^{\frac{q^{m}-q^{s}}{q-1}} + (q^{m}-q^{m-s})x^{\frac{q^{m-1} -q^{s-1}}{q-1}}y^{q^{m-1}-q^{s-1}} + (q^{m-s}-1)x^{\frac{q^{m-1}-q^{s}}{q-1}}y^{q^{m-1}}.
\]
\end{theorem}

\begin{proof}
Let $\vc \in \CC$ be a non-zero codeword of $\CC$ and $Y \subset [n]$ such that $\Ss(m) |_{[n]-Y} = \CC$, where $n$ is the length of $\Ss(m)$. Since both codes $\CC$ and $\Ss(m)$ have the same dimension, there is a bijection $\pi : \{ \vc': \vc'\in \Ss(m) \} \to \{ \vc : \vc \in \CC \}$ given by $\pi (\vc') = \vc'_{|_{[n]-Y}}$. Let $\vc'\in \Ss(m)$ be such that $\pi(\vc')=\vc$. We have
\[
\wt(\vc)=\wt(\pi(\vc'))=\wt(\vc'_{|_{[n]-Y}}) = \wt(\vc') - \wt(\vc'_{|_{Y}}).
\]
Let $H \in \HH(M_{\Ss})$ be the hyperplane obtained by $H=\varphi(\vc')$ in Lemma \ref{lemma:map_class_cdwd_hyperplanes}. Then, $\wt(\vc')=|\supp(\vc')| = n - |H|$ and $\wt(\vc'_{|_{Y}})=|Y| - |H \cap Y|$, since $H$ is the set of coordinates where $\vc'_{i}=0$ for $i \in H$. Hence, if $n_{\CC}$ denotes the length of $\CC$, we obtain
\[
\wt(\vc)=n_{\CC} - |H - H \cap Y|.
\] 
This shows that the weight of $\vc$ can be computed from the hyperplanes of $\Ss(m)$ and their relation with $Y$. By Lemma \ref{lemma:coatom_Y}, the hyperplanes of $M_{\Ss}$ split into two disjoint sets depending on whether they contain $Y$. We consider these two cases separately.

Suppose first that $Y \subset H$. Then, we have
\begin{align*}
\wt(\vc) &=n_{\CC} - |H| + |Y| \\
	& = \frac{q^{m}-q^{s}}{q-1} - \frac{q^{m-1}-1}{q-1} + \frac{q^{s}-1}{q-1} \\
	& = q^{m-1}.
\end{align*}
Now, the number of different hyperplane containing $Y$ in $M_{\Ss}$ is ${m-s \brack m-s-1}_{q} = \frac{q^{m-s}-1}{q-1}$. By Lemma \ref{lemma:map_class_cdwd_hyperplanes}, each hyperplanes yields a different equivalence class of codewords. Since there are $q-1$ codewords in each class, the number of codewords of weight $q^{m-1}$ is equal to $q^{m-s}-1$. 

Suppose now that $Y \nsubseteq H$. By Lemma \ref{lemma:coatom_Y} and the flats of the Simplex code, we know that $|H \cap Y| = \frac{q^{s-1}-1}{q-1}$. Then, we have
\begin{align*}
\wt(\vc) &=n_{\CC} - |H| + |H \cap Y| \\
	& = \frac{q^{m}-q^{s}}{q-1} - \frac{q^{m-1}-1}{q-1} + \frac{q^{s-1}-1}{q-1} \\
	& = q^{m-1} - q^{s-1}.
\end{align*}
The number of hyperplanes not containing $Y$ is ${m \brack m-1}_{q} - {m-s \brack m-s-1}_{q} = \frac{q^{m} - q^{s}}{q-1}$. By the same previous argument, the number of codewords of weight $q^{m-1}-q^{s-1}$ is then equal to $q^{m}-q^{s}$ and this concludes the proof. 
\qed
\end{proof}

%Recursive statement and example -----------------------------------
\subsection{The complete locality of $\Ss(m)-\Ss(s)$}

%Goal : get the other locality sets when <m-1 via recursive statement of Theorem. 
In this subsection, we describe the restrictions of $\CC=\Ss(m)-\Ss(s)$ of dimension less than $m-1$ to obtain the rest of the possible localities. The main observation is that Theorem \ref{thm:local_isom} was carefully written as a recursive statement on the same type of construction already considered, \ie, a Simplex code removed from another Simplex code. Therefore, we can apply Theorem \ref{thm:local_isom} again on the restriction $\CC|_{H}$ to obtain, up to isomorphism, the restricted codes of dimension $m-2$ of $\CC|_{H}$ and thus of $\CC$ as well. Moreover, since every code symbol is contained in a restriction of $\CC$ of dimension $m-1$, applying Theorem \ref{thm:local_isom} again on the restricted codes implies that every code symbol is also contained in a restriction of $\CC$ of dimension $m-2$. This is crucial when considering the locality of $\CC$. We start with an example.

%Small example like (12,4,6)
\begin{example}
\label{ex:1246_loc}
Let $\CC=\Ss_{2}(4)-\Ss_{2}(2)$ be the $[12,4,6]$ binary linear code of Example \ref{ex:1246}. Since $s<m-1$, Theorem \ref{thm:local_isom} implies that for all code symbols $e \in [n]$, there exist two sets $H_{1}, H_{2} \subseteq [12]$ containing $e$ such that $\CC|_{H_{1}}$ is isomorphic to $\Ss_{2}(3)-\Ss_{2}(1)$ and $\CC |_{H_{2}}$ is isomorphic to $\Ss_{2}(3)-\Ss_{2}(2)$. In other words, there are two restrictions $\CC|_{H_{1}}, \CC |_{H_{2}}$ containing $e$ with parameters $[6,3,3]$ and $[4,3,2]$ respectively. Therefore, $\CC$ is an LRC with locality $(r=4,\delta=3)$ and also an LRC with locality $(r=3,\delta=2)$. Notice that $r=4$ for the first locality even if the dimension of the restricted codes is equal to $3$. This is due to the fact that $r$ must satisfy $|R_{i}|\leq r+\delta-1$ as opposed to the definition of H-LRCs where $H(R_{i}) \leq r_{1}$. 

If we apply Theorem \ref{thm:local_isom} to $\Ss_{2}(3)-\Ss_{2}(1)$, we obtain by isomorphism that there exist $H_{3},H_{4} \subseteq H_{1}$ containing $e$ such that $\CC|_{H_{3}}$ is isomorphic to $\Ss_{2}(2)-\Ss_{2}(1)$ and $\CC|_{H_{4}}$ is isomorphic to $\Ss_{2}(2)$. $\Ss_{2}(2)-\Ss_{2}(1)$ is a $[2,2,1]$ code and thus does not provide an extra locality. On the other hand, $\Ss_{2}(2)$ is a $[3,2,2]$ code which implies that $\CC$ is also an LRC with locality $(r=2,\delta=2)$. Furthermore, by construction of the local sets, $\CC$ is an H-LRC with locality $[(3,3),(2,2)]$. 
\end{example}

%Goal : to the complete locality. Method and uniqueness-------------------
%Discuss where does it stop, left most and RM(1,m).
Example \ref{ex:1246_loc} illustrates how Theorem \ref{thm:local_isom} can be used to obtain the locality for different dimensions. Moreover, because of the recursive form of Theorem \ref{thm:local_isom}, the local sets can be arranged in such a way that we obtain a hierarchical locality. We break down what happens to the restrictions when we iterate Theorem \ref{thm:local_isom} by considering the restriction types, \ie , the different isomorphic restrictions. Suppose for simplicity that $m$ is sufficiently large and $s$ is close to half of $m$. As illustrated in Figure \ref{fig:Sm_triangle}, applying Theorem \ref{thm:local_isom} on the two restriction types of dimension $m-1$ gives three new restriction types of dimension $m-2$, as two of them lead to the same isomorphic code.

% Picture of triangle S(m)-S(s) - S(m-1)-S(s) / S(m-1)-S(s-1)
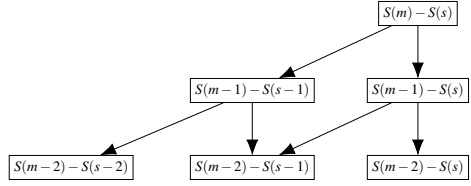
\begin{figure}
\centering
\resizebox{0.5\textwidth}{!}{%
	\begin{tikzpicture}
		\node[shape=rectangle,draw=black] (00) at (3.2,3) {\small $S(m)-S(s)$};
		%level 1
		\node[shape=rectangle,draw=black] (10) at (3.2,1.5) {\small $S(m-1)-S(s)$};
		\node[shape=rectangle,draw=black] (11) at (0,1.5) {\small $S(m-1)-S(s-1)$};

		%level 2
		\node[shape=rectangle,draw=black] (20) at (3.2,0) {\small $S(m-2)-S(s)$};
		\node[shape=rectangle,draw=black] (21) at (0,0) {\small $S(m-2)-S(s-1)$};
		\node[shape=rectangle,draw=black] (22) at (-3.5,0) {\small $S(m-2)-S(s-2)$};

     \path [-{Latex[length=3mm]}] (00) edge (10);
     \path [-{Latex[length=3mm]}] (00) edge (11);
     %1->2
     \path [-{Latex[length=3mm]}] (10) edge (20);
     \path [-{Latex[length=3mm]}] (10) edge (21);
      \path [-{Latex[length=3mm]}] (11) edge (21);
     \path [-{Latex[length=3mm]}] (11) edge (22);
     \end{tikzpicture}
     }
    \caption{Illustration of the top level locality of the code $\Ss(m)-\Ss(s)$ after appropriate isomorphisms.}
    \label{fig:Sm_triangle}
\end{figure}

Suppose now that, after some iterations of Theorem \ref{thm:local_isom}, we obtain the restriction types of dimension $\kappa$ . Let $a,b \in \Z_{+}$ such that all restriction types are of the form $\Ss(\kappa)-\Ss(i)$ with $a \leq i \leq b$. Now, the restriction types of dimension $\kappa -1$ can be obtained by applying Theorem \ref{thm:local_isom} on the restriction types of dimension $\kappa$. Two extremal cases need to be taken into account. If $a=0$, the only restriction type that we get from $\Ss(\kappa)-\Ss(a)$ is $\Ss(\kappa-1)$. If $b=\kappa-1$, the only restriction type is $\Ss(\kappa-1) - \Ss(\kappa-2)$. This is illustrated in Figure \ref{fig:proof_loc_kappa_scheme} where the two dashed boxes represent the conditional new restriction types that exist only if $a>0$ or $b<\kappa-1$. This sketches the high level idea of the proof of the next theorem describing precisely the different restriction types for a given dimension $\kappa$. 

\begin{figure}
\centering
\resizebox{0.7\textwidth}{!}{%
	\begin{tikzpicture}
	\node[shape=rectangle,draw=black,dashed] (00) at (0,0) {\small $S(\kappa-1)-S(a-1)$};
	\node[shape=rectangle,draw=black] (01) at (3.2,0) {\small $S(\kappa -1)-S(a)$};
	\node[] (20) at (5,0) { \small $\ldots$ };
	\node[shape=rectangle,draw=black] (03) at (7.2,0) {\small $S(\kappa -1)-S(b-1)$};
	\node[shape=rectangle,draw=black, dashed] (04) at (10.4,0) {\small $S(\kappa -1)-S(b)$};
	%upper part
	\node[shape=rectangle,draw=black] (10) at (3.2,1.5) {\small $S(\kappa)-S(a)$};
	\node[] (21) at (5,1.5) {\small $\ldots$};
	\node[shape=rectangle,draw=black] (12) at (7.2,1.5) {\small $S(\kappa)-S(b-1)$};
	\node[shape=rectangle,draw=black] (13) at (10.4,1.5) {\small $S(\kappa)-S(b)$};
	
	%vertical paths
     \path [-{Latex[length=3mm]}] (10) edge (01);
     \path [-{Latex[length=3mm]}] (12) edge (03);
     %left path
     \path [dashed,-{Latex[length=3mm]}] (10) edge node[above left, pos=0.5] {\small $a>0$ } (00);
     %right paths
     \path [dashed,-{Latex[length=3mm]}] (13) edge node[right] {\small $b<\kappa-1$} (04);
     \path [-{Latex[length=3mm]}] (13) edge (03);
     \end{tikzpicture}
     }
    \caption{Sketch of an iteration of Theorem \ref{thm:local_isom} on the different restrictions types.}
    \label{fig:proof_loc_kappa_scheme}
\end{figure}
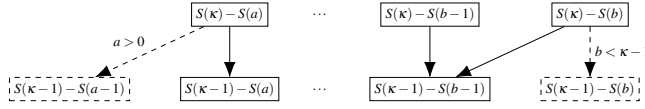

%Theorem about locality given the dimension. Locality for all kappa in between something and something.

\begin{theorem}
\label{thm:S_kappa}
Let $\CC=\Ss(m)-\Ss(s)$ with $m \geq 3$ and $ 0 \leq s \leq m-1$. Let $\kappa \in [2, m-1]$ and $i$ be an integer such that $\max \{0, s-m + \kappa \} \leq i \leq \min \{ s, \kappa -1 \}$. Then for all code symbols $e \in [n]$, there is a set $F_{i} \subseteq [n]$ containing $e$ such that $\CC|_{F_{i}}$ is isomorphic to $\Ss(\kappa)-\Ss(i)$. 
\end{theorem}

\begin{remark}
Notice that $\max \{0, s-m + \kappa \} \leq \min \{ s, \kappa -1 \}$. Indeed we have $\min \{ s, \kappa -1 \} \geq 0$ since $\kappa \geq 3$ and $s\geq 0$. Since $\kappa \leq m-1$ and $s \leq m-1$, then $s-m+\kappa \leq s-1$ and $s-m+\kappa \leq \kappa -1$. Thus, $\min \{ s, \kappa -1 \} \geq s-m+\kappa$ and we have that $\max \{0, s-m + \kappa \} \leq \min \{ s, \kappa -1 \}$. Therefore, the claim of Theorem \ref{thm:S_kappa} is that there always exist such restrictions for all dimensions $\kappa \in [2,m-1]$. 
\end{remark}

\begin{proof}
We denote by $I_{\kappa}$ the set $I_{\kappa} = \{ i \in \Z_{+} : \max \{0, s-m + \kappa \} \leq i \leq \min \{ s, \kappa -1 \} \}$. Let also $e \in [n]$ be an arbitrary code symbol.

The case $\kappa=m-1$ is exactly Theorem \ref{thm:local_isom} with the set $I_{m-1}$ being equal to $\{ 0 \}$ when $s=0$, $I_{m-1}= \{ m-2 \}$ if $s=m-1$, and $I_{m-1} = \{ s-1, s \}$ if $1 \leq s \leq m-2$. 

We now prove the theorem by doing a reverse induction on $\kappa$ as long as $\kappa \geq 3$. Assume that the claim is true for $\kappa$, \ie, for all $j \in I_{\kappa}$ and for all code symbol $e \in [n]$, there exists a set $W_{j} \subseteq [n]$ containing $e$ such that $\CC|_{W_{j}}$ is isomorphic to $\Ss(\kappa) - \Ss(j)$. The goal is to construct the restrictions of dimension $\kappa-1$. 

%First prove that the biggest index leads to another local set.---

Let $b=\min \{ s, \kappa -1 \}$. By the induction hypothesis, there exists a set $W_{b} \subseteq [n]$ containing $e$ such that $\CC|_{W_{b}}$ is isomorphic to $\Ss(\kappa) - \Ss(b)$. 

If $b=\kappa-1$ then applying Theorem \ref{thm:local_isom} on $\CC'= \Ss(\kappa)-\Ss(\kappa-1)$ yields that for all symbols $e'\in [n']$, with $n'$ being the length of the code $\CC'$, there is a set $F \subseteq [n']$ containing $e'$ such that $\CC'|_{F}$ is isomorphic to $\Ss(\kappa-1)-\Ss(\kappa-2)$. Therefore, there exists a set $F_{b-1} \subseteq [n]$ containing $e$ such that $\CC|_{F_{b-1}}$ is isomorphic to $\Ss(\kappa-1) - \Ss(\kappa-2)$ since we have that $\CC|_{F_{b-1}} = (\CC|_{W_{b}})|_{F_{b-1}}$. Since $b=\kappa-1$, then we have that $\kappa-2=b-1=\min \{ s, \kappa-2 \}$. 

If $b=s< \kappa -1$, then the third case of Theorem \ref{thm:local_isom} applied to $\Ss(\kappa) - \Ss(s)$ and the corresponding isomorphisms yields that there exists a set $F_{b}$ such that  $\CC|_{F_{b}}$ is isomorphic to $\Ss(\kappa) - \Ss(b)$. Moreover, since $b< \kappa -1$, then $b=\min \{ s, \kappa -2 \}$. 

% Creates the local sets just below. 

Let $j \in [\max \{0, s-m + \kappa \}, \min \{ s, \kappa -1 \} -1]$. Notice that this interval might be empty for example if $s=m-1$. If the interval is not empty, let $W_{j}$ the set containing $e$ by the induction hypothesis such that $\CC|_{W_{j}}$ is isomorphic to $\Ss(\kappa)- \Ss(j)$. Since by definition of the interval, we have that $j< \kappa -1$, we can apply Theorem \ref{thm:local_isom} to obtain a set $F_{i} \subseteq [n]$ containing $e$ such that $\CC|_{F_{i}}$ is isomorphic to $\Ss(\kappa - 1) - \Ss(j)$. 

Finally, if $\max \{0, s-m + \kappa \} > 0$, then let $a = s-m+\kappa$ and $W_{a}$ containing $e$ such that $\CC|_{W_{a}}$ is isomorphic to $\Ss(\kappa)-\Ss(a)$. Since $a \geq 1$, we can use the third case of Theorem \ref{thm:local_isom} applied to $\Ss(\kappa)-\Ss(a)$ to obtain a set $F_{a-1} \subseteq [n]$ such that $\CC|_{F_{a-1}}$ is isomorphic to $\Ss(\kappa-1) - \Ss(a-1)$. Moreover, we have $a-1 = \max \{ 0, s-m+\kappa -1 \}$. 

Hence, we have constructed a set for $e$ isomorphic to $\Ss(\kappa -1) - \Ss(i)$ for all $i$ in the set 

\begin{align*}
[\max & \{0, s-m + \kappa \}, \min \{ s, \kappa -1 \} -1] \\
& \qquad \cup \min \{ s, \kappa -2 \} \cup \max \{ 0, s-m+\kappa -1 \} \\
& =[\max \{ 0, s-m+\kappa -1 \}, \min \{ s, \kappa -2 \}] \\
& = I_{\kappa -1}.
\end{align*}

Since $e$ was an arbitrary symbol, this concludes the proof. 
\qed
\end{proof}

%To obtain locality, need to make sure that min dist is >1 and therefore cut it short. 

As a corollary of Theorem \ref{thm:S_kappa}, we have that $\CC$ is an LRC as long as the restriction types have a minimum distance greater than or equal to $2$. We choose to give here the length, dimension and minimum distance of the local codes to avoid confusion on the parameter $r$ and $r_{1}$ in the two definitions of LRCs and H-LRCs. 

\begin{corollary}
\label{cor:locality_kappa_params}
Let $\CC=\Ss(m)-\Ss(s)$ with $m \geq 3$ and $ 0 \leq s \leq m-1$. Let $\kappa \in [3, m-1]$ and $i$ be an integer such that $\max \{0, s-m + \kappa \} \leq i \leq \min \{ s, \kappa -1 \}$. Then $\CC$ is an LRC with local code parameters
\[
\begin{cases}
\left[ \frac{q^{\kappa} - q^{i}}{q-1}, \kappa , q^{\kappa -1} - q^{i-1}  \right] & \text{ if } i>0,\\
\left[ \frac{q^{\kappa} - 1}{q-1}, \kappa , q^{\kappa -1}  \right] & \text{ if } i=0. \\
\end{cases}
\] 

Furthermore, for $\kappa=2$, we have the following local parameters:
\begin{itemize}
\item If $0 \leq s \leq m-2$, then $\CC$ is an LRC with local parameters $[q+1,2,q]$
\item If $q>2$ and $1 \leq s \leq m-1$ then $\CC$ is an LRC with local parameters $[q,2,q-1]$.  
\end{itemize}
\end{corollary}

\begin{proof}
When $\kappa \geq 3$, Theorem \ref{thm:S_kappa} guaranties that for every code symbol $e \in [n]$, there is a restriction containing $e$ isomorphic to $\Ss(\kappa)-\Ss(i)$. Therefore $\CC$ is an LRC code with the parameters of the local sets given by the parameters of $\Ss(\kappa)-\Ss(i)$ with a minimum distance greater or equal than $2$. 

When $\kappa=2$, we need to distinguish some cases as the minimum distance might be smaller than $2$. Let $I_{2}= \{ i : \max \{0, s-m + 2 \} \leq i \leq \min \{ s, 1 \} \}$ be the range of the parameter $i$ when $\kappa=2$. If $s=0$, then $I_{2} = \{ 0 \}$ and by Theorem \ref{thm:S_kappa}, there is a restriction containing $e$ isomorphic to $\Ss(2)$. The parameters of the restriction is then $[q+1,2,q]$ and the minimum distance is thus greater or equal than $2$. 

If $1 \leq s \leq m-2$, then $s-m+2 \leq 0$ so $\max \{0, s-m + 2 \} =0$ and $I_{2}=\{0,1\}$. By Theorem \ref{thm:S_kappa}, every code symbol is contained in a restriction isomorphic to $\Ss(2)$ and in another restriction isomorphic to $\Ss(2)-\Ss(1)$. Since the parameters of $\Ss(2)-\Ss(1)$ are $[q,2,q-1]$, then the minimum distance is greater or equal than $2$ if and only if $q>2$. 

Finally, if $s=m-1$, then $s-m+2=1$ and $I_{2}=\{1\}$. Thus, $\CC$ is an LRC code with locality parameters $[q,2,q-1]$ if $q>2$. 
\qed
\end{proof}

%Special cases-----
% Uniqueness is guaranty by Proposition even after arbitrary set since every flat is contained in a flat of higher rank. 

Notice that when $\kappa=2$, we obtain again the locality proven in \cite{silberstein18}. We explain next why the list of localities of Corollary \ref{cor:locality_kappa_params} is the complete list of possible localities with closed local sets. In Corollary \ref{cor:locality_kappa_params}, the list of different localities is established by removing the restrictions leading to codes with minimum distance $1$ from the list of Theorem \ref{thm:S_kappa}. But the list of Theorem \ref{thm:S_kappa} relies on consecutive iterations of Theorem \ref{thm:local_isom}, where Proposition \ref{prop:coatom_isom} can be applied at each iteration guaranteeing the uniqueness of the restriction types. Therefore, Theorem \ref{thm:S_kappa} contains the complete list of restriction types and Corollary \ref{cor:locality_kappa_params} contains the complete list of possible localities with closed local sets. 

We will now look at two special cases of Theorem \ref{thm:S_kappa} and Corollary \ref{cor:locality_kappa_params}. In the first example, we study the case when $s=0$.

\begin{example}
Let $\CC = \Ss(m)-\Ss(0) = \Ss(m)$ with $m \geq 3$. Let $I_{\kappa}= \{ i : \max \{0, s-m + \kappa \} \leq i \leq \min \{ s, \kappa -1 \} \}$ be the range of the parameter $i$ in Theorem \ref{thm:S_kappa}. Since $s=0$, we have that $I_{\kappa}=0$. Then Theorem \ref{thm:S_kappa} and Corollary \ref{cor:locality_kappa_params} imply that $\CC$ is an LRC code with local parameters $\left[ \frac{q^{\kappa} - 1}{q-1}, \kappa , q^{\kappa -1}  \right]$ for all $\kappa \in [2,m-1]$. This result was already proven in \cite{grezet18LRC}. 
\end{example}

For the second example, we look at the case when $s=m-1$ and $\CC$ corresponds to the Reed--Muller code RM$(1,m-1)$.

\begin{example}
Let $\CC= \Ss(m)-\Ss(m-1)$ with $m \geq 3$ corresponding to the Reed--Muller code RM$(1,m-1)$. Let $I_{\kappa}$ the set $I_{\kappa} = \{ i : \max \{0, s-m + \kappa \} \leq i \leq \min \{ s, \kappa -1 \} \}$. Since $s=m-1$, we have $I_{\kappa}=\{ \kappa -1 \}$. Hence Theorem \ref{thm:S_kappa} and Corollary \ref{cor:locality_kappa_params} imply that RM$(1,m-1)$ is an LRC with local parameters $[q^{\kappa -1}, \kappa , q^{\kappa-1}-q^{\kappa -2}]$ for all $\kappa \in [3,m-1]$ and the local codes are isomorphic to the Reed--Muller code RM$(1,\kappa-1)$. Furthermore, if $q>2$, then RM$(1,m-1)$ is also an LRC with local parameters $[q,2,q-1]$. Moreover, since every locality is obtained by restrictions on the previous local sets, we get that the Reed--Muller codes RM$(1,m-1)$ are H-LRCs with $(m-3)$-level hierarchy over the binary field and with $(m-2)$-level hierarchy over $\F_{q}$ with $q>2$. 
\end{example}

We conclude this section by showing the optimality of $\Ss(m)-\Ss(s)$ with respect to the bounds \eqref{eq:CMG_r} and \eqref{eq:CM_HLRC_h} and present a table containing binary codes obtained by Construction \ref{cst:Sm_Ss} and their hierarchical localities. 

%Optimality for LRC, H-LRC with 2-level given by Griesmer bound and \lambda =0.  
\begin{theorem}
Let $\CC=\Ss(m)-\Ss(s)$ with $m \geq 3$ and $0 \leq s \leq m-1$. Then $\CC$ is an optimal LRC for every locality described in Corollary \ref{cor:locality_kappa_params}. Moreover, $\CC$ is an optimal H-LRC with
\[
\left\lbrace
\begin{array}{ll}
(m-2)\text{-level hierarchical locality } & \text{when } q>2 \text{ or } q=2 \text{ and } s<m-1,\\
(m-3)\text{-level hierarchical locality } & \text{when } q=2 \text{ and } s=m-1.\\
\end{array}
\right.
\]
\end{theorem}

\begin{proof}
Since $\CC$ already achieves the Griesmer bound on the parameters $[n,k,d]$, it will achieves the bound \eqref{eq:CMG_r} for $\lambda=0$ and thus be an optimal LRC code for each locality described in Corollary \ref{cor:locality_kappa_params}. 

We exhibit now a particular chain of subsets that yields the hierarchical locality. Vaguely speaking, the chain is obtained by taking the left foremost diagonal restriction types in Figures \ref{fig:Sm_triangle} and \ref{fig:proof_loc_kappa_scheme}. 

Formally, let $e \in [n]$ be an arbitrary code symbol. By Theorem \ref{thm:S_kappa}, for all $\kappa \in [2,m-1]$, there exists a set $F_{\kappa} \subseteq [n]$ containing $e$ such that $\CC|_{F_{\kappa}}$ is isomorphic to $\Ss(\kappa)-\Ss(\max\{0,s-m+\kappa\})$. Define $F_{m}=[n]$. Each $F_{\kappa}$ was obtained in the proof of Theorem \ref{thm:S_kappa} by applying Theorem \ref{thm:local_isom} to $\CC|_{F_{\kappa+1}}$. Therefore, we have $F_{2} \subseteq F_{3} \subseteq \cdots \subseteq F_{m-1} \subseteq [n]$. Moreover, because $\CC|_{F_{\kappa}}$ is isomorphic to $\Ss(\kappa)-\Ss(\max\{0,s-m+\kappa\})$, for every $a \in F_{\kappa}$, there exists such a chain $F_{2,a}\subseteq F_{3,a} \subseteq \cdots \subseteq F_{\kappa-1,a}$ with $a \in F_{l,a}$ and $\CC|_{F_{l,a}}$ is isomorphic to $\Ss(l)-\Ss(\max\{0,s-m+l\})$ for $l \in [2, \kappa-1]$. 

Now if $d$ is the minimum distance of $\CC$, then the minimum distance of $\Ss(\kappa)-\Ss(\max\{0,s-m+\kappa\})$ is 
\[
\delta_{\kappa} = \left\lbrace
\begin{array}{ll}
q^{\kappa-1}-q^{s-m+k-1} = \frac{d}{q^{m-\kappa}} & \text{if } \kappa > m-s, \\
q^{\kappa-1} & \text{if } \kappa \leq m-s. \\
\end{array}
\right.
\]
Notice that $\delta_{\kappa}$ is less than $2$ if and only if $q=2$, $s=m-1$, and $\kappa=2$ by Corollary \ref{cor:locality_kappa_params}. 

Hence, if $q>2$ or $q=2$ and $s<m-1$, then $\CC$ is an H-LRC with $(m-2)$-level hierarchical locality having locality parameters 
\[
[(m-1,\delta_{m-1}), \ldots , (2,\delta_{2})].
\]
When $q=2$ and $s=m-1$, $\CC$ is an H-LRC with $(m-3)$-level hierarchical locality having locality parameters 
\[
[(m-1, \delta_{m-1}), \ldots , (3, \delta_{3})].
\]

Finally, $\CC$ achieves the new alphabet-dependent bound \eqref{eq:CM_HLRC_h} for H-LRC codes when $\lambda=0$ and $k_{\opt}^{q}$ is taken to be the Griesmer bound. 
\qed
\end{proof}

The table in Figure \ref{fig:binary_Sm} represents the binary linear codes obtained by the construction $\Ss_{2}(m)-\Ss_{2}(s)$. The columns are sorted by $s \in [0,4]$ and the rows are sorted by $m \in [2,6]$. Moreover, the lines describes the locality of dimension $m-1$ obtained by Theorem \ref{thm:local_isom}. Therefore, if $\CC$ and $\CC'$ are two codes in the table such that there exists a path from $\CC'$ to $\CC$, then $\CC$ has locality $\CC'$, \ie, each symbol of $\CC$ is contained in a restriction isomorphic to $\CC'$. Figure \ref{fig:binary_Sm} gives also the hierarchical locality of a binary code via the paths to smaller codes. Finally, the codes in blue are the binary Reed--Muller codes RM$(1,m-1)$. 

%Full table for binary codes. 

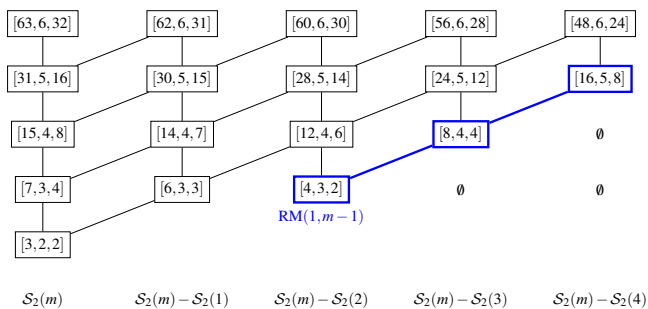
\begin{figure}
\centering
\resizebox{0.7\textwidth}{!}{%
\begin{tikzpicture}
	%Simplex S(m)
	\node (60) at (0,-2) {\small $\Ss_{2}(m)$};
	\node[shape=rectangle,draw=black] (50) at (0,-1) {\small $[3,2,2]$};
	\node[shape=rectangle,draw=black] (00) at (0,0) {\small $[7,3,4]$};
	\node[shape=rectangle,draw=black] (10) at (0,1) {\small $[15,4,8]$};
	\node[shape=rectangle,draw=black] (20) at (0,2) {\small $[31,5,16]$};
	\node[shape=rectangle,draw=black] (30) at (0,3) {\small $[63,6,32]$};
	
	 \path [-] (50) edge (00);
     \path [-] (00) edge (10);
     \path [-] (10) edge (20);
     \path [-] (20) edge (30);

     %Code S(m)-S(1)
     \node (61) at (2.5,-2) {\small $\Ss_{2}(m) - \Ss_{2}(1)$};
	\node[shape=rectangle,draw=black] (01) at (2.5,0) {\small $[6,3,3]$};
	\node[shape=rectangle,draw=black] (11) at (2.5,1) {\small $[14,4,7]$};
	\node[shape=rectangle,draw=black] (21) at (2.5,2) {\small $[30,5,15]$};
	\node[shape=rectangle,draw=black] (31) at (2.5,3) {\small $[62,6,31]$};
	
     \path [-] (01) edge (11);
     \path [-] (11) edge (21);
     \path [-] (21) edge (31);
     
    %Code S(m)-S(2)
    \node (62) at (5,-2) {\small $\Ss_{2}(m)-\Ss_{2}(2)$};
	\node[shape=rectangle,draw=blue, very thick, label=below:{\textcolor{blue}{RM$(1,m-1)$}}] (02) at (5,0) {\small $[4,3,2]$};
	\node[shape=rectangle,draw=black] (12) at (5,1) {\small $[12,4,6]$};
	\node[shape=rectangle,draw=black] (22) at (5,2) {\small $[28,5,14]$};
	\node[shape=rectangle,draw=black] (32) at (5,3) {\small $[60,6,30]$};
	
     \path [-] (02) edge (12);
     \path [-] (12) edge (22);
     \path [-] (22) edge (32);

	%Code S(m)-S(3)
	\node (63) at (7.5,-2) {\small $\Ss_{2}(m)-\Ss_{2}(3)$};
	\node (03) at (7.5,0) {\small $\emptyset$};
	\node[shape=rectangle,draw=blue, very thick] (13) at (7.5,1) {\small $[8,4,4]$};
	\node[shape=rectangle,draw=black] (23) at (7.5,2) {\small $[24,5,12]$};
	\node[shape=rectangle,draw=black] (33) at (7.5,3) {\small $[56,6,28]$};
	
     %\path [-] (03) edge (13);
     \path [-] (13) edge (23);
     \path [-] (23) edge (33);

	%Code S(m)-S(4)
	\node (64) at (10,-2) {\small $\Ss_{2}(m)-\Ss_{2}(4)$};
	\node (04) at (10,0) {\small $\emptyset$};
	\node (14) at (10,1) {\small $\emptyset$};
	\node[shape=rectangle,draw=blue, very thick] (24) at (10,2) {\small $[16,5,8]$};
	\node[shape=rectangle,draw=black] (34) at (10,3) {\small $[48,6,24]$};
	
     %\path [-] (01) edge (11);
     %\path [-] (13) edge (23);
     \path [-] (24) edge (34);

	%Diagonal part     
	\path [-] (50) edge (01);
   	\path [-] (00) edge (11);
    \path [-] (10) edge (21);
    \path [-] (20) edge (31);

    \path [-] (01) edge (12);
    \path [-] (11) edge (22);
    \path [-] (21) edge (32);
     
    \path [-,color=blue, very thick] (02) edge (13);
    \path [-] (12) edge (23);
    \path [-] (22) edge (33);

    \path [-,color=blue, very thick] (13) edge (24);
    \path [-] (23) edge (34);

\end{tikzpicture}
}
\caption{Table of the complete hierarchical locality for binary codes $\Ss_{2}(m)-\Ss_{2}(s)$ with $m \in [2,6]$ and $s \in [0,4]$.}
\label{fig:binary_Sm}
\end{figure}

\section{Conclusion}
In this paper, we presented a new alphabet-dependent bound for codes with hierarchical locality. Then, we worked on a class of codes obtained by deleting a Simplex code from another Simplex code of higher dimension. We derived the weight enumerator of these codes and the complete list of possible localities with closed repair sets. Finally, we used this list to show that these codes are optimal LRCs and optimal H-LRCs by the new bound.

\begin{acknowledgements}
This work was supported in part by the Academy of Finland, under grants \#276031, \#282938, and \#303819, and by the Technical University of Munich -- Institute for Advanced Study, funded by the German Excellence Initiative and the EU 7th Framework Programme under grant agreement \#291763, via a Hans Fischer Fellowship.
\end{acknowledgements}

\bibliographystyle{spmpsci}      % mathematics and physical sciences
\bibliography{references}

\end{document}